%% file: main.tex
\documentclass[10pt, conference, letterpaper]{IEEEtran}

\usepackage{caption}
\usepackage{algpseudocode}
\usepackage{amsmath}
\usepackage{amsfonts}
\usepackage{amssymb}
\usepackage{subfigure}

\usepackage{graphicx}
\usepackage{array,multirow}

\usepackage{multicol}

\usepackage{amsfonts}
\usepackage{booktabs}
\usepackage{siunitx}

\usepackage{color}

\usepackage{textgreek}

\usepackage{tikz}

\usepackage{enumitem}

\usepackage{url}
\makeatletter
\def\url@leostyle{%
  \@ifundefined{selectfont}{\def\UrlFont{\sf}}{\def\UrlFont{\small\ttfamily}}}
\makeatother
\newcommand{\eat}[1]{}
\usepackage[linesnumbered,ruled]{algorithm2e}
\usepackage{mdframed} 

\usepackage{makecell}

\usepackage{mathtools}

\usepackage{xcolor}
\definecolor{light-gray}{gray}{0.9}

\newenvironment{packed_enum}{%
  \begin{enumerate}%
  }{\end{enumerate}}

\usepackage{amsthm}
\theoremstyle{definition}
\newtheorem{definition}{Definition}
\newtheorem{theorem}{Theorem}
\newtheorem{lemma}{Lemma}

\newcolumntype{L}[1]{>{\raggedright\let\newline\\\arraybackslash\hspace{0pt}}m{#1}}

\usepackage{listings}

\definecolor{codegreen}{rgb}{0,0.6,0}
\definecolor{codegray}{rgb}{0.5,0.5,0.5}
\definecolor{codepurple}{rgb}{0.58,0,0.82}
\definecolor{backcolour}{rgb}{0.95,0.95,0.92}
 
\lstdefinestyle{mystyle}{
    backgroundcolor=\color{backcolour},   
    commentstyle=\color{codegreen},
    keywordstyle=\color{magenta},
    numberstyle=\tiny\color{codegray},
    stringstyle=\color{codepurple},
    basicstyle=\footnotesize,
    breakatwhitespace=false,         
    breaklines=true,                 
    captionpos=b,                    
    keepspaces=true,                 
    numbers=left,                    
    numbersep=5pt,                  
    showspaces=false,                
    showstringspaces=false,
    showtabs=false,                  
    tabsize=2
}
 
\usepackage{arydshln}

\input{solidity-highlighting.tex}

\begin{document}

\title{NF-Crowd: Nearly-free Blockchain-based Crowdsourcing}



\author{
  \IEEEauthorblockN{
    Chao~Li\IEEEauthorrefmark{1},
    Balaji~Palanisamy\IEEEauthorrefmark{2}, 
    Runhua~Xu\IEEEauthorrefmark{2},
    Jian~Wang\IEEEauthorrefmark{1} and
    Jiqiang~Liu\IEEEauthorrefmark{1}
  }
  \IEEEauthorblockA{
    \IEEEauthorrefmark{1}Beijing Key Laboratory of Security and Privacy in Intelligent Transportation, Beijing Jiaotong University, Beijing, China \\
    \IEEEauthorrefmark{2}School of Computing and Information, University of Pittsburgh, Pittsburgh, USA
  }
}

\IEEEtitleabstractindextext{%

\begin{abstract}

Advancements in distributed ledger technologies are rapidly driving the rise of decentralized crowdsourcing systems on top of open smart contract platforms like Ethereum. While decentralized blockchain-based crowdsourcing provides numerous benefits compared to centralized solutions, current implementations of decentralized crowdsourcing suffer from 
fundamental scalability limitations
by requiring all participants to pay a small transaction fee every time they interact with the blockchain. This increases the cost of using decentralized crowdsourcing solutions, resulting in a total payment that could be even higher than the price charged by centralized crowdsourcing platforms.
This paper proposes a novel suite of protocols called \texttt{NF-Crowd} that 
resolves the scalability issue
by reducing the lower bound of the total cost of a decentralized crowdsourcing project to $O(1)$. 
\texttt{NF-Crowd} is a highly reliable solution for scaling decentralized crowdsourcing.
We prove that as long as participants of a project powered by \texttt{NF-Crowd} are rational, the $O(1)$ lower bound of cost could be reached regardless of the scale of the crowd.
We also demonstrate that as long as at least one participant of a project powered by \texttt{NF-Crowd} is honest, the project cannot be aborted and the results are guaranteed to be correct.
We design \texttt{NF-Crowd} protocols for a representative type of project named crowdsourcing contest with open community review (CC-OCR).
We implement the protocols over the Ethereum official test network.
Our results demonstrate that \texttt{NF-Crowd} protocols can reduce the cost of running a CC-OCR project to less than \$2 regardless of the scale of the crowd, providing a significant cost benefit in adopting decentralized crowdsourcing solutions.

\end{abstract}


}

\maketitle

\IEEEdisplaynontitleabstractindextext

\IEEEpeerreviewmaketitle

\section{Introduction}

In the recent years, crowdsourcing has been gaining attention as a promising modern business model that enables individuals and organizations to receive services from a large group of people or crowd. Crowdsourcing services support a variety of tasks ranging from software development to logo designs~\cite{chittilappilly2016survey}.
For example, \texttt{LEGO Ideas} is attracting a lot of fan designers to enter prize contests by submitting original proposals for new \texttt{LEGO} Ideas sets. A few top-ranked proposals in the contest have resulted in successful commercialization~\cite{LEGO}.
Likewise, since 2005, UNIQLO has continuously held annual Global T-Shirt Design Competitions (UTGP) and its 14th UTGP in 2019 has  attracted over 18,000 entries from all over the world~\cite{UTGP}.
The increasing popularity of crowdsourcing solutions is also driving the rise of crowdsourcing intermediate platforms such as \texttt{Upwork}~\cite{Upwork}, \texttt{99designs}~\cite{99designs} and \texttt{designContest}~\cite{designContest} that connect business clients to designers including individual freelancers and design agencies.
For instance, via \texttt{designContest}, a startup company may set up a logo design contest with a description of its requirements and include a monetary reward.
After receiving a large number of entries from interested designers, the client may invite its target audience such as its followers on Twitter to vote for the favorite design and finally pick one or multiple winning entries based on the voting result.
Nearly all such crowdsourcing intermediate platforms make profits by charging fees from clients and designers.
\texttt{Upwork} reported in second quarter 2019 that its Gross Services Volume (GSV) grew 20\% year-over-year to \$518.8 million~\cite{UIR}.
On a \$500 crowdsourcing project in \texttt{Upwork}, the platform would charge \$100 as the service fee.
Similarly, \texttt{99designs} would charge \$75 from a \$500 project as the platform fee.
Such high service fees significantly increase the cost of running crowdsourcing projects online. Clients and designers have no choice but accept it as the mutually distrusted parties need a trustworthy intermediary for avoiding dishonest behaviors such as free-riding and false-reporting~\cite{zhang2012reputation}.

Recent advancements in blockchain
technology have led to the development of numerous open smart contract platforms including Ethereum~\cite{buterin2014next,wood2014ethereum}.
Ethereum has become a promising technology for decentralizing traditional centralized online services and as result, in the recent years, we have witnessed a rapid proliferation of numerous decentralized applications including decentralized crowdsourcing systems~\cite{calado2018tamper,duan2019aggregating,feng2019mcs,li2018crowdbc,lu2018zebralancer,wang2018blockchain,xu2019blockchain}. Such services
offer clients and designers an option to reduce the high intermediary fee required in centralized crowdsourcing systems.
However, blockchains provide tamper resistance properties only at a cost. For example, Ethereum charges each transaction a small fee based on the complexity. In Ethereum, tens of thousands of miners follow the Proof-of-Work (PoW) consensus protocol~\cite{nakamoto2008bitcoin} to compete for solving puzzles and each winner receives a monetary reward for packaging the recent transactions  (i.e., transferring fund, executing functions or creating smart contracts) into a new block appended to the end of the blockchain. Fees charged from transactions within a new block in the blockchain are paid to the competition winner who packages the block.
People trust that no one can tamper with the blockchain as the probability of a single miner to win in several consecutive competitions 
to be able to change the network consensus about the blockchain state is negligible.
As part of the competition reward, transaction fees help incentivize miners to invest more computation resources into the competitions, which in turn increases the difficulty of competitions and improves the overall safety of the blockchain.
Also, transaction fees help protect Ethereum against DDoS and Sybil attacks~\cite{douceur2002sybil} as the cost of creating $n$ transactions will have a cost of $O(n)$.
Despite their significance to the safety of Ethereum, transaction fees turn out to be an obstacle to existing decentralized crowdsourcing systems, 
especially when crowdsourcing projects are scaled up.
Consider that a client sets up a decentralized design contest in Ethereum, where each entry needs to be submitted by a designer via a transaction that costs a small fee, say \$1.
The total fee charged in this contest would be cheaper than a contest in \texttt{99designs} only when there are less than 75 entries.
In other words, intermediary fees are not eliminated by decentralizing crowdsourcing but are paid to a decentralized infrastructure instead of a centralized service provider for the same purpose of acquiring trust.
As a result, it is not surprising to see that sometimes there may be no considerable economic advantage in decentralizing crowdsourcing.
For example, if we take the price of ether~\footnote{\begin{scriptsize} The native cryptocurrency in Ethereum, denoted by \textXi. \end{scriptsize}} 
as its mean value during the first half of the year 2019 recorded in \textit{Etherscan}~\cite{etherscan}, 
a crowdsourced image tagging task completed via CrowdBC~\cite{li2018crowdbc} could very well spend over four times the price charged by the centralized Amazon Mechanical Turk~\cite{AMT}. Similarly, aggregating data from 1,000 providers via decentralized crowdsensing~\cite{duan2019aggregating} could cost up to \$170.

This paper aims at addressing the challenging question: 
\textit{how to provide a reliable solution that decouples the cost of decentralizing crowdsourcing from the scale of the crowd?}
Our research illustrates that the root cause of the high cost in existing decentralized crowdsourcing systems is the lack of  cost-efficient solutions to exploit decentralized trust.
Crowdsourcing projects usually involve steps that aggregate data (i.e., contest entries or sensed data) from the crowd of scale $n$ or perform calculations on aggregated data (i.e., determine winning entries based on votes).
Such steps become cost-intensive in smart contract platforms like Ethereum as they either charge accumulated small transaction fees (TYPE $n \times 1$) or a single large transaction fee (TYPE $1 \times n$), both resulting in $O(n)$ cost.
In this paper, we propose a novel suite of protocols called \texttt{NF-Crowd} that 
reliably resolves the scalability issue by
reducing the lower bound of the total cost of a decentralized crowdsourcing project to $O(1)$.
We prove that as long as participants of a project powered by \texttt{NF-Crowd} are rational, the $O(1)$ lower bound of cost could be reached regardless of the scale of the crowd.
We also demonstrate that as long as at least one participant of a project powered by \texttt{NF-Crowd} is honest, the project cannot be aborted and the results are guaranteed to be correct.
We design \texttt{NF-Crowd} protocols for a representative type of project named crowdsourcing contest with open community review (CC-OCR).
We implement the protocols over the Ethereum official test network.
Our results demonstrate that \texttt{NF-Crowd} protocols can reduce the cost of running a CC-OCR project to less than \$2 regardless of the scale of the crowd, providing a significant cost benefit in adopting decentralized crowdsourcing solutions.

The rest of this paper is organized as follows: 
We start by introducing preliminaries in Section~\ref{s2}.
In Section~\ref{s3}, we present a strawman protocol for CC-OCR projects and categorize the cost-intensive steps resulting in $O(n)$ cost.
Then, in Section~\ref{s4}, we propose the \texttt{NF-Crowd} protocol for CC-OCR projects that reduces the lower bound of cost to $O(1)$.
We implement and evaluate the \texttt{NF-Crowd} protocols over the Ethereum official test network in Section~\ref{s5}.
Finally, we discuss related work in Section~\ref{s6} and conclude in Section~\ref{s7}.

\section{Preliminaries}
\label{s2}
In this section, we discuss the preliminaries about smart contracts and introduce the key assumptions, key cryptographic tools and notations used in our work.
Though we discuss smart contracts in the context of Ethereum~\cite{wood2014ethereum}, our solutions are applicable to a wide range of other smart contract platforms as well.

\subsection{Account types}
\label{s2.1}
There are two types of accounts in Ethereum, namely External Owned Accounts (EOAs) and Contract Accounts (CAs).
To interact with the Ethereum blockchain, a user needs to create an EOA and control it via a pair of keys.
Specifically, the public key can generate a 20-byte address to uniquely identify the EOA and the private key can be used by the user to sign transactions or other types of messages.
Then, any user can create a smart contract by sending out a contract creation transaction from a controlled EOA.
The 20-byte address of the created smart contract 
becomes the unique identity of the contract account (CA).

\subsection{Transactions}
\label{s2.2}
The state of Ethereum blockchain can only be changed by the external world (i.e., EOAs) using transactions.
A transaction is a serialized binary message sent from an EOA (i.e., sender) that contains the following elements:
\begin{itemize}
  \item \textit{nonce}: a sequence number issued by the EOA (i.e., transaction creator) to prevent transaction replay;
  \item \textit{gas price}: the price of gas the EOA is willing to pay;
  \item \textit{gas limit}: the maximum amount of gas the EOA can afford;
  \item \textit{recipient}: the recipient account address;
  \item \textit{value}: the amount of ether to send to the recipient;
  \item \textit{data}: the binary data payload;
  \item \textit{vrs}: the ECDSA digital signature of the EOA.
\end{itemize}

Depending on the value at \textit{recipient} (i.e., EOA or CA or 0x0), transactions can be classified into three categories.

\noindent \textbf{Fund transfer transaction}:
A transaction with an EOA as \textit{recipient} and a non-empty \textit{value} is a fund transfer transaction, which is used to transfer an amount of ether from the sender EOA to the \textit{recipient} EOA.
On the other hand, \textit{data} carried by a transaction is usually ignored by Ethereum clients and wallets that help users control their EOAs.

\noindent \textbf{Function invocation transaction}: 
When a transaction involves a CA as \textit{recipient} as well as a non-empty \textit{data}, it is usually a function invocation transaction for calling a function within an existing smart contract.
Specifically, when the transaction is used to call a function with arguments, such as:

\begin{minipage}{\linewidth}
\begin{lstlisting}[
linewidth=8.5cm,
language=Solidity,
basicstyle=\footnotesize,
label={a2},
frame=none,
numbers=none,
numbersep=5pt,
breaklines=true,
breakatwhitespace=true
]
  function f(uint _arg1, uint _arg2) public {}
\end{lstlisting}
\end{minipage}

\noindent the \textit{data} payload of the transaction would be in the form of
$$encode(`f(uint256,uint256)')|encode(\_arg1)|encode(\_arg2)$$
namely concatenation of the encoded string $``f(uint256,uint256)"$, also called as \textit{function selector}, and the encoded values for all function arguments.
When a function invocation transaction additionally carries a non-empty \textit{value} and the invoked function is marked \textit{payable} such as:

\begin{minipage}{\linewidth}
\begin{lstlisting}[
linewidth=8.5cm,
language=Solidity,
basicstyle=\footnotesize,
label={a2},
frame=none,
numbers=none,
numbersep=5pt,
breaklines=true,
breakatwhitespace=true
]
function withdraw(uint _amount) public payable {}
\end{lstlisting}
\end{minipage}

\noindent The amount of ether indicated by \textit{value} would be transferred from the sender EOA to the \textit{recipient} CA.

\noindent \textbf{Contract creation transaction}: 
In Ethereum, there is a special type of transaction for creating new smart contracts.
Such a transaction, usually referred to as a contract creation transaction, carries a special \textit{recipient} address 0x0, an empty \textit{value} and a non-empty \textit{data} payload.
A smart contract (or contract) in Ethereum is a piece of program created using a high-level contract-oriented programming language such as \textit{Solidity}~\cite{Solidity2017}.
After compiling into a low-level bytecode language called Ethereum Virtual Machine (EVM) code, the created contract is filled into a contract creation transaction as the \textit{data} payload.

A user, after filling \textit{recipient}, \textit{value} and \textit{data} of a transaction, will then input \textit{nonce}, \textit{gas price} and \textit{gas limit} and finally sign the transaction with the private key of the sender EOA to get signature \textit{vrs}.
After that, a complete transaction is created.
To make the transaction get executed to change the state of the Ethereum blockchain, the transaction should be broadcast to the entire Ethereum network formed by tens of thousands of miner nodes.
Following the Proof-of-Work (PoW) consensus protocol~\cite{nakamoto2008bitcoin}, miners in Ethereum competitively solve a blockchain puzzle and the winner packages the received transactions into a block and appends the new block to the end of Ethereum blockchain.
From then on, it is hard to tamper with the blockchain state updated by the transaction (i.e., transferred fund, executed function or created contract).
Thus, transactions and smart contracts in Ethereum are executed transparently in a decentralized manner and the results are deterministic.


\subsection{Transaction fees}
\label{s3.3}
In order to either deploy a new contract or call a deployed contract in Ethereum, one needs to spend Gas, or transaction fees.
Based on the complexity of the contract or that of the called function, an amount of ether needs to be spent to purchase an amount of Gas as a transaction fee, which is then paid to the winning miner. 
The Gas system is important for Ethereum as it helps to incentivize miners to stay honest, to nullify denial-of-service attacks and to encourage efficiency in smart contract programming.
On the other hand, the Gas system requires protocols, especially the multi-party ones, to be designed with higher scalability in Ethereum. This is due to the fact that even a single-round multi-party protocol could spend a lot of money to run in case of a large number of participants. 

\subsection{Off-chain channels}
\label{s3.4}
In Ethereum, nodes forming the underlying P2P network can send messages to each other via off-chain channels established through the Whisper protocol~\cite{Whisper2017}.
By default, messages are broadcast to the entire P2P network.
A node can set up a filter to only accept messages marked with a specific 4-byte topic.
Besides, a node can locally generate a pair of asymmetric Whisper keys and reveal the public Whisper key to the blockchain, which allows other nodes to privately communicate with it.

\subsection{Key assumptions}
\label{assumptions}
\vspace{-1mm}
We make the following key assumptions in this paper:
\begin{itemize}[leftmargin=*]
\item We assume that the underlying blockchain system satisfies \textit{liveness}, \textit{consistency} and \textit{immutability} properties~\cite{garay2017bitcoin}. 
\item We also make the synchrony assumption that ensures that there exists a known upper bound on the delay of messages.
\item Finally, we assume that at least one reviewer is honest. 
\end{itemize}

\begin{table}
\caption{Summary of notations.}
\vspace{-3.5mm}    
\begin{center}
\begin{tabular}{|c|p{6.5cm}|}
\hline
\textbf{notation} & \textbf{description} \\
\hline
$C$ & a client who sets up a design contest \\
$D$ & a designer who wants to win a reward\\
$R$ & a reviewer who casts vote to entries\\
$S$ & a smart contract \\
$S.f()$ & function $f()$ within contract $S$ \\
$\rightrightarrows$ & broadcast information via off-chain channels\\
$\dashrightarrow$ & transmit infomation via \textit{private} off-chain channels \\
$\Rightarrow$ & invoke a function within a smart contract\\
$\textit{addr}\;(*)$ & an address of an EOA or a CA\\
$\textit{keccak}\;(*)$ & a keccak hash value\\
$\textit{ipfs}\;(*)$ & a content-addressed IPFS link\\
\hline
\end{tabular}
\end{center}  
\vspace{-4.5mm}
\label{t1}
\end{table}

\subsection{Cryptographic tools}
\label{s3.5}
The design of \texttt{NF-Crowd} protocols employs several key cryptographic tools:
(1) we use a standard notion of Merkle trees and denote the function of getting the Merkle root of a tree as $root \gets merkleRoot(elements)$.
(2) we use the Keccak 256-bit hash function supported by Ethereum and it is denoted as $keccak(*)$.
(3) we use the ECDSA signature supported by Ethereum. 
Specifically, a EOA (i.e., \textit{signer}) can sign any message via \textit{JavaScript} API and get signature $vrs \gets sig(keccak(message))$.
Later, other EOAs or CAs can recover the address of the \textit{signer} EOA (i.e. $addr(signer)$) via \textit{JavaScript} API or \textit{Solidity} native function and get
$addr(signer) \gets vf(keccak(message),vrs)$.
(4) we use a distributed file system named InterPlanetary File System (IPFS)~\cite{benet2014ipfs} to cost-efficiently store data to Ethereum blockchain.
Specifically, instead of paying a large transaction fee to store a large file on the chain, we could just store the immutable, permanent IPFS link (hash) of the file obtained via $link \gets \textit{ipfs}\;(\textit{file})$ by paying a much smaller fee, which still timestamps and secures the content of file.

In the next two sections, we start by presenting a strawman protocol for a CC-OCR crowdsourcing scenario. We then introduce the idea of \texttt{NF-Crowd} for reducing cost to $O(1)$. 
Before discussing the proposed protocols, we summarize the notations that will be used in the rest of this paper in Table~\ref{t1}.

\section{CC-OCR: A strawman protocol}
\label{s3}
In this section, we first describe the crowdsourcing contest with open community review (CC-OCR) as a three-phase process. We then propose a strawman protocol that decentralizes CC-OCR projects over Ethereum. We finally identify and categorize the cost-intensive steps within the strawman protocol that results in its $O(n)$ total cost.

\subsection{CC-OCR as a three-phase process}
We describe a CC-OCR project as a three-phase process similar to the procedure of \texttt{LEGO Ideas}~\cite{LEGO}:

\begin{itemize}[leftmargin=*]

\item \textit{CC-OCR.initial}: 
The client ($C$) initiates a design contest with a detailed description of its requirements as well as a reward. 
\item \textit{CC-OCR.entry}:
The interested designers ($D$s) within the community submit their entries to the contest.
\item \textit{CC-OCR.review}:
The client ($C$) invites selected members from the community as reviewers ($R$s) to cast votes to the entries (one vote per $R$) and pick one or multiple winning entries based on the voting result.
Here, the known identities of reviewers prevent them from performing a Sybil attack and hence, the one vote per reviewer could be guaranteed.

\end{itemize}

\begin{figure}
\centering
{
   
    \includegraphics[width=9cm,height=3.5cm]{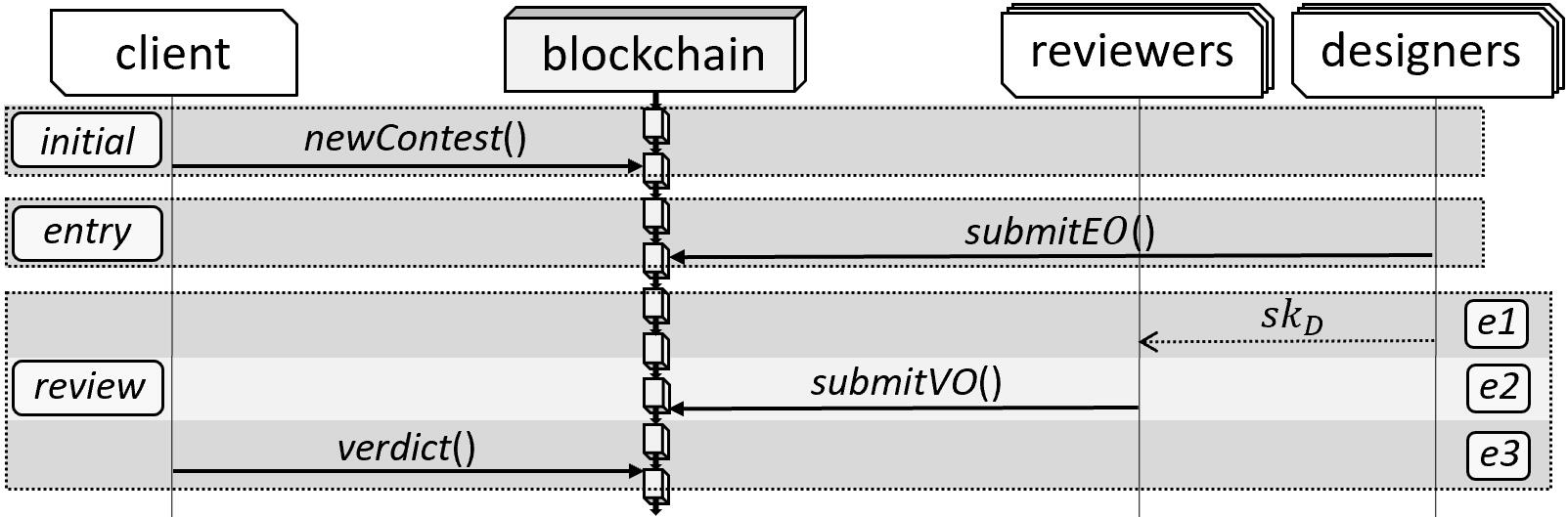}
}
\caption {Strawman CC-OCR protocol sketch. Solid lines denote on-chain transactions. Dotted lines denote off-chain communication.}
\label{protocol_sketch_01} 
\end{figure}

\subsection{The strawman CC-OCR protocol}

We first propose a strawman protocol to decentralize CC-OCR projects over Ethereum.
We sketch the protocol in Fig.~\ref{protocol_sketch_01} and present the formal description in Fig.~\ref{protocol_detail_1}.
The regulations of the strawman protocol are programmed as an agency smart contract $S_{\text{ag}}$, through which a client can hold a contest to receive entries and votes from a community.

\noindent \textbf{\textit{CC-OCR.initial}}:
Client $C$ initiates a contest by sending out a function invocation transaction from an owned EOA to call the function $newContest()$ at CA $addr(S_{\text{ag}})$.
The transaction would carry two arguments at its $data$ payload: 
(1) a list of $deadlines$ indicating ending times of phases or epochs in the contest;
(2) an IPFS link guiding designers to a file describing the requirements of this contest.
The transaction would also include a non-empty $value$ to transfer an amount of ether to CA $addr(S_{\text{ag}})$ to reward winning entries.
From then on, the arguments are permanently recorded in the blockchain and this new contest can be uniquely identified via $addr(C)$ the client address and $cn$ the number of contests that have been created by $C$ at CA $addr(S_{\text{ag}})$.

\begin{figure}
\begin{minipage}{0.5\textwidth}
\begin{mdframed}[innerleftmargin=8pt]

\noindent \textbf{CC-OCR.initial:} 
\begin{packed_enum}[leftmargin=*]
  \item Client creates a contest: $C \Rightarrow S_{\text{ag}}.newContest$ 
    $(deadlines,  
    \textit{ipfs}\;(description)$, \textXi $reward$).
\end{packed_enum}

\noindent \textbf{CC-OCR.entry:} 
\begin{packed_enum}[leftmargin=*]
  \setcounter{enumi}{1}
  \item Designers $Ds$ submit entry objects $(EOs)$:
  \begin{packed_enum}[leftmargin=*]
    \item Each $D$ creates  
    $EO \coloneqq \textit{ipfs}\;(E(pk_D,proposal))$. 
    \item $D \Rightarrow S_{\text{ag}}$.
      $submitEO([addr(C),cn],EO)$. 
  \end{packed_enum}
\end{packed_enum}

\noindent \textbf{CC-OCR.review:} 
\begin{packed_enum}[leftmargin=*]

  \setcounter{enumi}{2}
  \item Each $D \rightrightarrows$:
    $sk_D$.

  \item Reviewers $Rs$ submit vote objects $(VOs)$:
  \begin{packed_enum}[leftmargin=*]
    \item Each $R$ creates 
      $VO \coloneqq addr(D)$.
    \item $R \Rightarrow S_{\text{ag}}$.
      $submitVO([addr(C),cn],VO)$.
  \end{packed_enum}

  \item Client reveals final verdict: 
  $C \Rightarrow S_{\text{ag}}$. $verdict(cn)$.

\end{packed_enum}

\end{mdframed}
\end{minipage}

\captionof{figure}{
Strawman CC-OCR protocol
}
\label{protocol_detail_1}
\end{figure}

\noindent \textbf{\textit{CC-OCR.entry}}: 
After reading the descriptions, interested designers from the community could submit entries to the contest.
The protocol requires a designer $D$ to create an entry object \textit{EO} by first generating a one-time asymmetric key pair, then encrypting the detailed $proposal$ with the public key $pk_D$ and finally computing the IPFS link of encrypted $proposal$.
Such an \textit{EO} could make $proposal$ confidential before the review phase and also reduce the size of data stored on the chain.
After that, the designer could call $submitEO()$ with a transaction carrying arguments \textit{EO} as well as $[addr(C),cn]$ that specifies the participating contest.
It is worth noting that any entry submmitted after the $deadline$ specified by $C$ for the entry phase would be rejected because $submitEntry()$ includes a time checker function:

\begin{minipage}{\linewidth}
\begin{lstlisting}[
linewidth=8cm,
language=Solidity,
basicstyle=\footnotesize,
label={a2},
frame=none,
numbers=none,
numbersep=5pt,
breaklines=true,
breakatwhitespace=true
]
               require(now < deadline);
\end{lstlisting}
\end{minipage}

 \noindent which requires miners to only accept the transaction if the timestamp is smaller than $deadline$ input via $newContest()$.

\noindent \textbf{\textit{CC-OCR.review}}: 
The protocol divides the review phase into three epochs.
The first epoch is for designers to reveal private keys $sk_D$ so that reviewers can read their entries.
Then, during the second epoch, each reviewer $R$ creates a vote object \textit{VO} which in the strawman protocol is simply the address of the designer that $R$ wants to vote for.
The created \textit{VO} would be submitted to CA $addr(S_{\text{ag}})$ via $submitVO()$.
Finally, client $C$ shall call $verdict()$ during the third epoch and the function would traverse the votes received by all the entries and pick one or multiple entries obtaining the greatest number of votes as winners, who could later withdraw the \textXi $reward$ deposited by $C$ via $newContest()$.

\subsection{The cost-intensive steps}

Despite the simplicity, the strawman CC-OCR protocol involves two key representative types of cost-intensive steps, defined as TYPE $n \times 1$ and TYPE $1 \times n$, respectively.

\noindent \begin{definition}[\textsc{TYPE $n \times 1$}]
\textit{A protocol step is said to be in TYPE $n \times 1$ if the number of transactions created at this step increases along with the scale of the crowd.}
\end{definition}

\noindent \begin{definition}[\textsc{TYPE $1 \times n$}]
\textit{A protocol step is said to be in TYPE $1 \times n$ if the cost of a single transaction created at this step increases along with the scale of the crowd.}
\end{definition}

   

\noindent 
It is easy to see that both step 2 and 4 of the strawman protocol in Fig.~\ref{protocol_detail_1} are in TYPE $n \times 1$ while step 5 is in TYPE $1 \times n$.
At step 2 (4), data is aggregated from the crowd and each designer (reviewer) needs to submit data via a function invocation transaction $submitEO()$ ($submitVO()$) that spends a certain amount of ether and hence, the accumulated cost at this step increases along with the scale of the crowd (i.e., number of transactions).
At step 5, function $verdict()$ contains a \textit{for} loop to iterate over all entries to pick the winners and each iteration spends a certain amount of ether and therefore, the accumulated cost at this step also increases along with the scale of the crowd (i.e., number of iterations).

\noindent \begin{theorem}
\textit{A TYPE $1 \times n$ step is associated with at least one TYPE $n \times 1$ step, but not vice versa.}
\end{theorem}

\noindent 
We can see that even though steps 2 and 4 are in TYPE $n \times 1$, there is a difference between them.
Specifically, step 2 leverages the blockchain as a bulletin board to post entries that meet the requirements of the contest (e.g., before $deadline$) in a deterministic and trustworthy way and hence, data received at step 2 (i.e., \textit{EO}) is not associated with any future step in the protocol.
In contrast, data aggregated at step 4 (i.e., \textit{VO}) is associated with step 5 as the votes are collected for the purpose of being the inputs of a function that outputs poll winners.

Next, we introduce the NF-crowd CC-OCR protocol that leverages cost-cutting strategies to convert the two types of cost-intensive steps into cost-efficient steps and reduce the lower bound of the total cost of a decentralized CC-OCR project to $O(1)$.

\section{CC-OCR: An NF-Crowd protocol}
\label{s4}
We begin by presenting the high-level ideas of the NF-Crowd CC-OCR protocol.
We then present the strategies of cutting costs at TYPE $n \times 1$ steps and TYPE $1 \times n$ steps in detail.
We finally analyze the total cost and safety of the proposed protocol. 
We present the NF-Crowd CC-OCR protocol sketch in Fig.~\ref{sketch_02} and the formal description in Fig.~\ref{protocol_detail_2}.

\subsection{High-level ideas}
The NF-Crowd CC-OCR protocol design includes the following novel mechanisms:

\noindent \textbf{\textit{Enforceable off-chain execution}}:
Similar to the strawman CC-OCR protocol, most existing decentralized crowdsourcing systems are designed to be executed in an on-chain mode, where both data storage and computation over stored data are performed on the blockchain to employ the decentralized trust in a simple but expensive way~\cite{duan2019aggregating,li2018crowdbc,wu2019bptm}.
We believe that decentralized trust could be employed in a more cost-effective way.
Specifically, we design the NF-Crowd CC-OCR protocol to be executed in an off-chain mode by default, involving no cost-intensive on-chain operations resulting in $O(n)$ cost when all participants are honest.
In case if any dishonest participant performs any fraudulent behavior that violates the protocol and aborts the off-chain execution,
any honest participant reserves the ability to switch the off-chain mode to an on-chain mode similar to the strawman CC-OCR protocol so that the execution of the protocol could always get enforced.

\noindent \textbf{\textit{Punishable protocol violation}}:
In order to incentivize participants to stay honest so that the protocol can end successfully in its off-chain mode, we require each participant to lock an amount of ether in smart contracts as a security deposit to penalize potential misbehaviors that violate the protocol by confiscating the security deposit paid by the corresponding dishonest participant.
Specifically, anyone who wishes to join the community to engage in a contest needs to first send out a function invocation transaction to call a function called $joinCommunity()$ at CA $addr(S_{\text{ag}})$. The transaction carries no arguments but is with a non-empty $value$ to transfer an amount of ether to the CA as \textXi $deposit$.
Likewise, during the initial phase, besides \textXi $reward$, the $newContest()$ function also charges an amount of ether from client $C$ as \textXi $deposit$.


\subsection{TYPE $n \times 1$ cost-cutting strategy}

The proposed strategy for cutting the cost of TYPE $n \times 1$ steps consists of four components, illustrated as step 2.1 (4.1), step 2.2 (4.2), step 2.3 (4.3) and step 6 in Fig.~\ref{protocol_detail_2}, respectively.

\noindent \textbf{\textit{Off-chain aggregation}}:
The first component converts TYPE $n \times 1$ on-chain aggregation into TYPE $1 \times n$ off-chain aggregation. 
The on-chain aggregation employed in the strawman protocol requires participants (i.e., designers or reviewers) to separately submit data to the blockchain, resulting in $n$ transactions.
In the NF-Crowd protocol, participants transmit data to a single uploader (i.e., client $C$), who then submits aggregated data via a single transaction to the blockchain.
In Fig.~\ref{protocol_detail_2}, at step 2.1 (4.1), each \textit{EO} (\textit{VO}) additionally carries a signature $vrs_D$ ($vrs_R$) and is then transmitted to client $C$ via off-chain channels.

\begin{figure}
\centering
{
   
    \includegraphics[width=9cm,height=6.5cm]{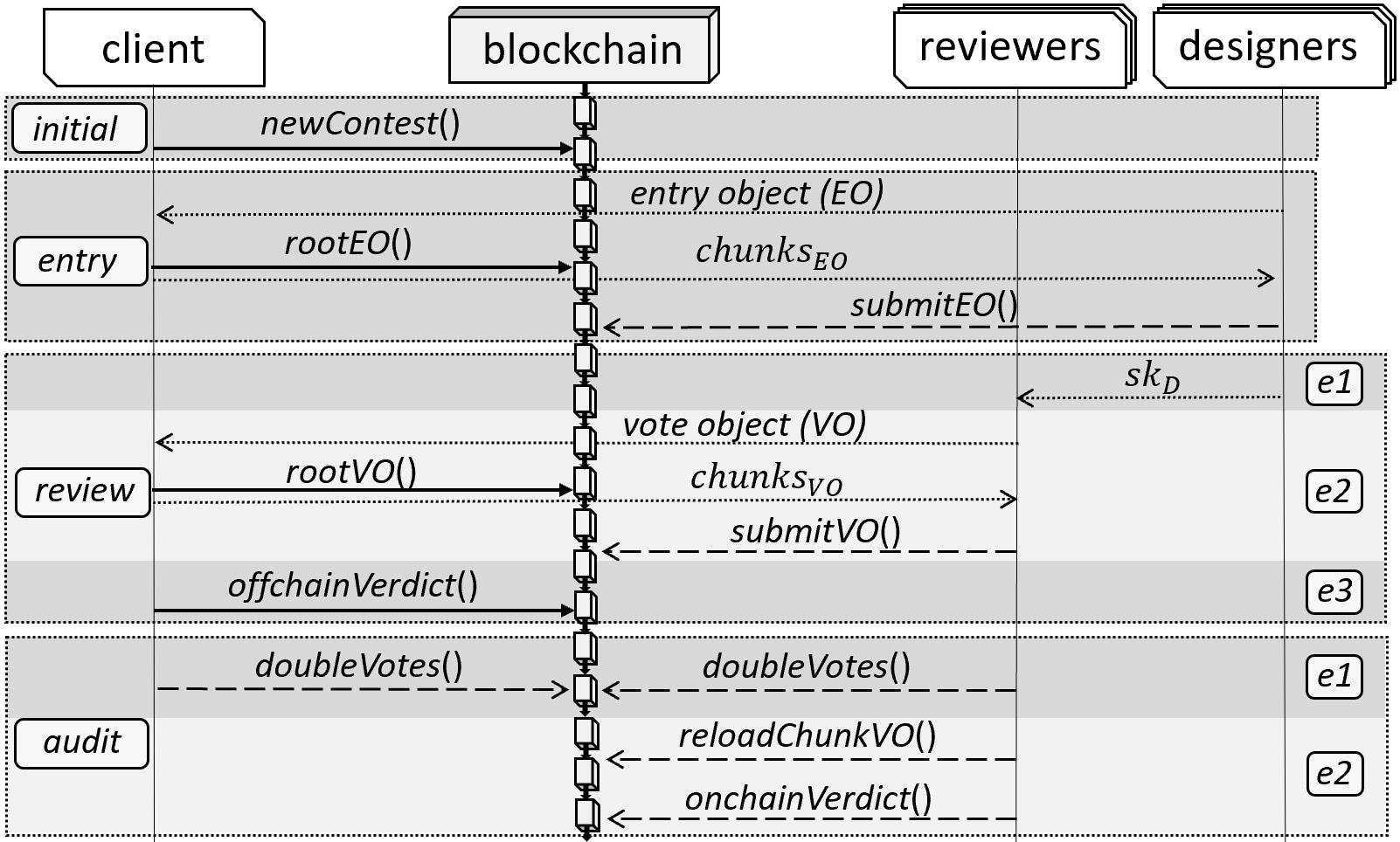}
}
\caption {NF-Crowd CC-OCR protocol sketch. Solid lines denote on-chain transactions. Dotted lines denote off-chain communication. Dashed lines denote available on-chain transactions as countermeasures.}
\vspace{-3mm}
\label{sketch_02} 
\end{figure}

\begin{figure}
\begin{minipage}{0.51\textwidth}
\begin{mdframed}[innerleftmargin=2pt]

\noindent \textbf{CC-OCR.initial:} 
\begin{packed_enum}[leftmargin=*]
  \item Client creates a contest: $C \Rightarrow S_{\text{ag}}.newContest$ 
    $(deadlines, $ 
    $\textit{ipfs}\;(description)$, [\textXi $deposit$, \textXi $reward$]).
\end{packed_enum}

\noindent \textbf{CC-OCR.entry:} 
\begin{packed_enum}[leftmargin=*]
  \setcounter{enumi}{1}
  \item Designers $Ds$ submit entry objects $(EOs)$:
  \begin{packed_enum}[leftmargin=*]
    \item $D$ creates 
      $EO \coloneqq [vrs_D,\textit{ipfs}\;(E(pk_D,proposal))]$, \\
      where $vrs_D \gets sig(keccak(\textit{ipfs}\;(E(pk_D,proposal))))$. \\
      Then, $D \dashrightarrow C$: 
      $[EO]$.
    \item $C$ organizes received $EOs$ as a Merkle tree and computes $root_{EO} \gets merkleRoot(chunks_{EO})$:
    \begin{packed_enum}[leftmargin=*]
      \item $C \Rightarrow S_{\text{ag}}$. 
        $rootEO(cn,root_{EO})$.
      \item  $C \rightrightarrows$:
        $\textit{ipfs}\;(chunks_{EO})$.
    \end{packed_enum}
    \item $<Countermeasure\ against\ intentional\ exclusion>$ \\
      Designer $D \Rightarrow S_{\text{ag}}$.$submitEO([addr(C),cn],EO)$. \\
      The transaction fee shall be shared by $C$ and $D$.
  \end{packed_enum}
\end{packed_enum}

\noindent \textbf{CC-OCR.review:} 
\begin{packed_enum}[leftmargin=*]

  \setcounter{enumi}{2}
  \item Each $D \rightrightarrows$:
    $sk_D$.

  \item Reviewers $Rs$ submit vote objects $(\textit{VOs})$:
  \begin{packed_enum}[leftmargin=*]
    \item $R$ creates 
      $VO \coloneqq [vrs_R,addr(D)]$, where $vrs_R \gets sig(keccak(addr(D)))$. 
      Then, $R \dashrightarrow C$: 
      $[\textit{VO}]$.
    \item $C$ organizes received \textit{VOs} as a Merkle tree and computes $root_{\textit{VO}} \gets merkleRoot(chunks_{\textit{VO}})$:
    \begin{packed_enum}[leftmargin=*]
      \item $C \Rightarrow S_{\text{ag}}$. 
        $\textit{rootVO}(cn,root_{\textit{VO}})$.
      \item  $C \rightrightarrows$:
        $\textit{ipfs}\;(chunks_{\textit{VO}})$.
    \end{packed_enum}
    \item $<Countermeasure\ against\ intentional\ exclusion>$ \\
      Reviewer $R \Rightarrow S_{\text{ag}}$.$\textit{submitVO}([addr(C),cn],\textit{VO})$. \\
      The transaction fee shall be shared by $C$ and $R$.
  \end{packed_enum}

  \item $C \Rightarrow S_{\text{ag}}$.$\textit{offChainVerdict}(cn, winners)$.

\end{packed_enum}

\noindent \textbf{CC-OCR.audit:} 
\begin{packed_enum}[leftmargin=*]
  \setcounter{enumi}{5}
  \item $<Countermeasure\ against\ double\ votes>$ \\
    $C$ (or $R$) $\Rightarrow S_{\text{ag}}$.$\textit{doubleVotes}(addr(C),cn,\textit{proof},chunk_{\textit{VO}},i)$.
  \item $<Countermeasure\ against\ incorrect$ \textit{off-chain} $computation>$ \\
    $R \Rightarrow S_{\text{ag}}$.$\textit{reloadChunkVO}(addr(C),cn,\textit{proof},chunk_{\textit{VO}})$.
    $R \Rightarrow S_{\text{ag}}$.$\textit{onchainVerdict}(addr(C),cn)$.
\end{packed_enum}

\end{mdframed}
\end{minipage}

\captionof{figure}{
NF-Crowd CC-OCR protocol
}
\vspace{-3mm}
\label{protocol_detail_2}
\end{figure}

\noindent \textbf{\textit{On-chain Merkle root}}:
The second component leverages a Merkle tree to reduce the cost of TYPE $1 \times n$ off-chain aggregation from $O(n)$ to $O(1)$.
Instead of uploading raw aggregated data to the blockchain, client $C$ here can group the aggregated data as a number of chunks\footnote{\begin{scriptsize} Aggregated data is grouped into chunks for reducing the cost of reloading aggregated data onto the chain. We discuss more details on this in Section~\ref{s5}. \end{scriptsize}}, create a Merkle tree that takes the chunks as elements, make chunks public via off-chain channels and upload the Merkle root to the blockchain.
This strategy has the following properties:

\begin{itemize}[leftmargin=*]
\item \textit{O(1) cost}:
The only data that needs to be uploaded is the Merkle root and the cost of uploading the 32-byte root is constant.
\item \textit{Transparency}: 
All chunks are available to all participants via off-chain channels.
\item \textit{Tamper resistance}: 
The integrity of each chunk could be verified via the on-chain Merkle root and the integrity of each \textit{EO (VO)} could be verified via $vrs_D$ ($vrs_R$).
\item \textit{Traceability}: 
Signatures $vrs_D$ ($vrs_R$) could be used to obtain the addresses of the designers (reviewers).
\end{itemize}

Thus, as long as all the participants are honest, this proposed strategy can achieve $O(1)$ cost with the same level of security as in the on-chain aggregation strategy.
At step 2.2 (4.2) in Fig.~\ref{protocol_detail_2}, client $C$ creates a Merkle tree for received $\textit{EOs}$ ($\textit{VOs}$), upload Merkle root via function $\textit{rootEO}()$ ($\textit{rootVO}()$) and reveal the IPFS link of $chunks_{EO}$ ($chunks_{\textit{VO}}$) via off-chain channels.

\noindent \textbf{\textit{Countermeasure against intentional exclusion}}:
In on-chain aggregation, data is uploaded by participants themselves and it is hard to keep any participant from getting involved in a contest.
In off-chain aggregation, for the purpose of cutting costs, data from all participants is uploaded together by a single uploader, hence a dishonest uploader here has the ability to intentionally exclude data belonging to certain participants from the Merkle tree.
As a countermeasure strategy, when a designer (reviewer) fails to verify her data via the on-chain Merkle root, 
the designer (reviewer) could re-upload the data onto the chain by herself as in the on-chain aggregation before the end of entry (review.e2) phase.
The additional cost is shared between the participant and the uploader.
In other words, we design the off-chain aggregation strategy to be backed up by the on-chain aggregation strategy and hence a dishonest uploader can hardly block any participant.
Meanwhile, both the uploader and participants are incentivized to honestly follow the off-chain strategy to avoid additional charges in the on-chain strategy.
This component is illustrated at step 2.3 (4.3) in Fig.~\ref{protocol_detail_2}, where function $submitEO()$ ($\textit{submitVO}()$) would refund half of the transaction fee to the designer (reviewer) from \textXi $deposit$ paid by client $C$.

\noindent \textbf{\textit{Countermeasure against double votes}}:
It is possible that a dishonest participant uploads data twice, first time via off-chain aggregation and second time via on-chain aggregation.
For instance, a dishonest reviewer can first submit a \textit{VO} to client at step 4.1 and later submit the same \textit{VO} at step 4.3, even if the \textit{VO} could be correctly verified through the on-chain Merkle root.
Without taking care of this, the vote may be counted twice, which violates the `per vote per reviewer' rule.
It is hard for contract $S_{\text{ag}}$ to detect the double-vote misbehavior because the contract has no knowledge of the off-chain chunks.
Even if $S_{\text{ag}}$ knows all chunks, it would be quite expensive to verify the double-vote misbehavior on the chain.
Therefore, we design a $\textit{doubleVotes}()$ function that could be called by either client $C$ or any reviewer $R$ to detect such a misbehavior off the chain and report it with a \textit{proof} to $S_{\text{ag}}$ during epoch-1 of the audit phase that follows the review phase so that the misbehavior could be efficiently verified by $S_{\text{ag}}$ on the chain.
We show the pseudo-code of $\textit{doubleVotes}()$ at Algorithm~\ref{A1}.
The function first verifies the timestamp of the transaction is within epoch-1 of the audit phase (line 1).
Then, it computes the hash of the input chunk (i.e., \textit{chunk}) that contains the repeated \textit{VO}, retrieves the Merkle root uploaded by client $C$ at step 4.2 and verifies the input Merkle proof (i.e., \textit{proof}) (line 2-4).
After that, function $splitChunk()$ would retrieve the repeated \textit{VO} from the chunk (i.e., $i^{th}$ \textit{VO} in the chunk) and decompose that \textit{VO} to $vrs_R$ and $addr(D)$ (line 5), from which the contract gets the address of $R$ who voted $D$ via off-chain aggregation (line 6).
Finally, if the contract finds that $R$ has also voted at step 4.3, it will mark $R$ as dishonest and record the address of the reporter and remove the votes cast by $R$ from the poll.
Later, the reporter could withdraw an award confiscated from \textXi $deposit$ paid by the violator.

\begin{algorithm}[t]
    \footnotesize
    \SetKwInOut{Input}{Input}
    \SetKwInOut{Output}{Output}

    \Input{$C, cn, proof, chunk, i.$}
    $verifyTimestamp(now\ is\ in\ audit.e1)$\;
    $hash \gets keccak256(chunk)$\;
    $root \gets retrieveRoot(C,cn)$\;
    \If{$verifyMerkleProof(proof,root,hash) == TRUE$} 
    {
      $(v, r, s, D) \gets splitChunk(chunk, i)$\;
      $R \gets vf(keccak256(D), v, r, s)$\;
    }

    \caption{The doubleVotes() function}
    \label{A1} 
\end{algorithm}


\subsection{TYPE $1 \times n$ cost-cutting strategy}

Our strategy of cutting the cost of TYPE $1 \times n$ steps consists of two components, illustrated as step 5 and step 7 in Fig.~\ref{protocol_detail_2}, respectively.

\noindent \textbf{\textit{Off-chain computation}}:
By default, the NF-Crowd protocol encourages computations to be performed off the chain.
Therefore at step 5 in Fig.~\ref{protocol_detail_2}, client $C$ could compute the winners off the chain after combining \textit{VOs} received at step 4.2 and \textit{VOs} uploaded at step 4.3 (if any) and simply upload the addresses of winning designers onto the chain via function $\textit{offChainVerdict()}$.
If no one challenges the results in a certain period of time, the contract would take the results as the final verdict.
In this way, the expensive on-chain computation could be offloaded to the off-chain side and the cost of uploading computation results is usually a small constant value.

\noindent \textbf{\textit{Countermeasure against incorrect off-chain computation}}:
It is important that the offloaded computation could be reloaded onto the blockchain so that off-chain computation results could be replaced with trustworthy on-chain computation results at any moment.
By always backing up off-chain computation with on-chain computation, incorrect off-chain computation results would never be adopted in the final verdict.
The procedure of reloading off-chain computation consists of two steps. 
First, as presented in Theorem~1, a TYPE $1 \times n$ step is associated with at least one TYPE $n \times 1$ step, so that any data offloaded via the TYPE $n \times 1$ cost-cutting strategy presented in the Section~4.2 needs to be reloaded onto the chain.
After that, computations could be performed on the reloaded data just as in the strawman protocol.
For instance, at step 7 in Fig.~\ref{protocol_detail_2}, a reviewer $R$ decides to challenge $winners$ uploaded by client $C$ at step 5.
To do this, the reporter should first reload all $chunks_{\textit{VO}}$ onto the chain via function $\textit{reloadChunkVO}()$.
We show the pseudo-code of $\textit{reloadChunkVO}()$ at Algorithm~\ref{A2}.
The function first verifies the timestamp (line 1) and Merkle proof (line 2-4).
Then, the function runs a \textit{for} loop to traverse \textit{VOs} included in the input chunk (line 5-9), recover the address of reviewer signing each \textit{VO} (line 6-7) and reload all the votes inside the chunk onto the chain (line 8).
Then, the computation could be re-done on the chain via function $\textit{onchainVerdict}()$.
If the results of $\textit{onchainVerdict}()$ are different from the ones in \textit{offchainVerdict()}, the final verdict would take the results of $\textit{onchainVerdict}()$ and a part of \textXi $deposit$ paid by client $C$ would be transferred to the reporter as an award.
Otherwise, the final verdict would still take the results of $\textit{offchainVerdict()}$.

\begin{algorithm}[t]
    \footnotesize
    \SetKwInOut{Input}{Input}
    \SetKwInOut{Output}{Output}

    \Input{$C, cn, proof, chunk.$}
    $verifyTimestamp(now\ is\ in\ audit.e2)$\;
    $hash \gets keccak256(chunk)$\;
    $root \gets retrieveRoot(C,cn)$\;
    \If{$verifyMerkleProof(proof,root,hash) == TRUE$} 
    {
      \For{$i=0;i<size(chunk);i++$}
      {
        $(v, r, s, D) \gets splitChunk(chunk, i)$\;
        $R \gets vf(keccak256(D), v, r, s)$\;
        $reload(R,D)$\;
      }
    }

    \caption{The reloadChunkVO() function}
    \label{A2} 
\end{algorithm}

\subsection{Cost and Security analysis}

We analyze the cost and security of the proposed NF-Crowd CC-OCR protocol as follows:

\begin{lemma}
\label{lemma1}
The NF-Crowd CC-OCR protocol has a total cost in the range of $[O(1),O(n)]$ and the $O(1)$ lower bound could be reached as long as the participants are rational.
\end{lemma}

\begin{proof}
The total cost would reach the upper bound $c_{\textit{ocr}}^{\textit{up}}$ when 
(1) all $\textit{EOs}$ are uploaded via $submitEO()$ onto the blockchain and 
(2) all $\textit{VOs}$ are uploaded via $\textit{submitVO}()$ or $\textit{reloadChunkVO}()$ onto the blockchain and 
(3) function $\textit{doubleVotes}()$ is called for each \textit{VO} and
(4) function \textit{onchainVerdict()} is invoked, namely,
$$O(c_{\textit{ocr}}^{\textit{up}}) \to O(c_{\textit{eo}}\cdot p_{\textit{d}}\cdot n+(c_{\textit{vo}}+c_{\textit{dv}})\cdot p_{\textit{r}} \cdot n+c_{\textit{ov}}+c_{\textit{rest}}) \to O(n),$$
where $c_{\textit{eo}}$ and $c_{\textit{vo}}$ denotes cost of uploading per \textit{EO} and \textit{VO}, $p_{\textit{d}}$ and $p_{\textit{r}}$ denotes percentage of designers and reviewers in the crowd that engage in the contest, $c_{\textit{dv}}$ and $c_{\textit{ov}}$ express cost of $\textit{doubleVotes}()$ and \textit{onchainVerdict()} and finally $c_{\textit{rest}}$ represents the total cost of calling other functions inside $S_{\text{ag}}$.
In contrast, the total cost would reach the lower bound $c_{\textit{ocr}}^{\textit{low}}$ when none of $\textit{submitEO}()$,$\textit{submitVO}()$, $\textit{reloadChunkVO}()$, $\textit{doubleVotes}()$ and \textit{onchainVerdict()} have been invoked, namely
$O(c_{\textit{ocr}}^{\textit{low}}) \to O(c_{\textit{rest}}) \to O(1)$.
All the five functions are countermeasures against dishonest participants by fixing problems made by them and confiscating \textXi $deposit$ paid by them, so the violators would gain no positive benefit but only a negative payoff.
Considering that rational adversaries choose to violate protocols only when doing so brings them a positive payoff~\cite{dong2017betrayal,groce2012fair,guo2016rational,nguyen2013analyzing}, rational participants would not choose to lose \textXi $deposit$, which in turn would push the total cost to reach its lower bound.
\end{proof}

We define the abortion of protocols to be the case that neither \textit{offchainVerdict()} nor \textit{onchainVerdict()} has been executed by the end of audit phase.
We define the correction of results to be the case that the effect of all the three identified protocol violations (\textit{intentional exclusion}, \textit{double vote} and \textit{incorrect off-chain computation}) to the results has been eliminated by the end of audit phase. 

\begin{lemma}
\label{lemma2}
The NF-Crowd CC-OCR protocol cannot be aborted and the results are guaranteed to be correct as long as at least one reviewer is honest.
\end{lemma}

\begin{proof}
We have assumed that at least one reviewer is honest in Seection~\ref{assumptions}. Therefore, in the worst case, the client and all but one reviewers are dishonest. 
We first prove the never-abort property. 
During the review phase, 
the client may refuse to call \textit{offchainVerdict()}, resulting in \textit{incorrect off-chain computation}.
As a countermeasure, the honest reviewer could always reload all votes via \textit{reloadChunkVO()} and pick the winners via \textit{onchainVerdict()}.
Therefore, with the existence of a single honest reviewer, it is guaranteed that either \textit{offChainVerdict()} or \textit{onchainVerdict()} will be called by the end of audit phase.
Next, we prove the always-correct property. As long as there is a single honest reviewer, this reviewer always reserves the ability to include her vote into the pool via \textit{submitVO()} (\textit{intentional exclusion}), to eliminate all illegal on-chain votes via \textit{doubleVotes()} (\textit{double vote}), to reload all legal off-chain votes onto the chain via \textit{reloadChunkVO()} and finally to pick the winners from the pool formed by all legal votes via \textit{onchainVerdict()} (\textit{incorrect off-chain computation}). Therefore, the effect of all the three identified protocol violations to the results can be eliminated before end of audit phase.
\end{proof}

\section{Implementation and evaluation}
\label{s5}
In this section, we present the evaluation of the two protocols: 
strawman CC-OCR and NF-Crowd CC-OCR.
We programmed a smart contract $S_{\text{agency}}$ for each of the protocols using \textit{Solidity} and evaluated the protocols over the Ethereum official test network \textit{kovan}~\cite{kovan}.
Similar to recent work on blockchain-based platforms~\cite{das2018yoda,dziembowski2019perun}, the key focus of our evaluation is on measuring gas consumption as the execution complexity and monetary cost in Ethereum are measured via gas consumption.
In addition, we compare the cost of the proposed protocols with that of two existing solutions.

\subsection{Gas consumption of proposed protocols}

In Table~\ref{t2}, we list the key functions in the programmed contracts that interact with protocol participants during different phases and the cost of these functions in both Gas and USD.
For ease of expression, a function marked $\ast$, $\bullet$ or $^{\ast}_{\bullet}$ represents the function is included in strawman CC-OCR, NF-Crowd CC-OCR or both the two protocols, respectively.
The cost in USD is computed through 
$cost(\textit{USD})$=$cost(Gas)*\textit{GasToEther}*\textit{EtherToUSD}$,
where $\textit{GasToEther}$ and $\textit{EtherToUSD}$ are taken as their mean value during the first half of the year 2019 recorded in \textit{Etherscan}~\cite{etherscan}, 
which are $1.67*10^{-8}$ Ether/Gas and 175 USD/Ether, respectively.
We next evaluate the cost of the two protocols:

\noindent \textbf{\textit{Strawman CC-OCR protocol}}: 
The cost in Strawman CC-OCR includes the following components:
(1) a client to set up a new contest via $newContest()$ (\$0.53) during \textit{CC-OCR.initial};
(2) each interested designer to submit an entry via $submitEO()$ (\$0.42) during \textit{CC-OCR.entry};
(3) each reviewer to cast a vote via $\textit{submitVO}()$ (\$0.18) during \textit{CC-OCR.review} 
and 
(4) client to pick winners via $verdict()$.
It thus costs about $\$(0.53+0.42n_d+0.18n_r+(0.11+0.006n_d)) = \$(0.64+0.426n_d+0.18n_r)$ for completing a contest that involves $n_d$ designers and $n_r$ reviewers, where $\$(0.11+0.006n_d)$ is the average cost of picking winners among $n_d$ designers in $verdict()$.

\begin{table}
\caption{Key functions and their cost in Gas and USD. A function marked $\ast$, $\bullet$ or $^{\ast}_{\bullet}$ represents the function is included in strawman CC-OCR, NF-Crowd CC-OCR or both the two protocols, respectively.}      
\begin{center}
\begin{tabular}{|p{6mm}|p{18mm}|p{21mm}|p{12mm}|p{10mm}|}
\hline
    \textbf{Phase} & \textbf{Function} & \textbf{Description} & \textbf{Gas} & \textbf{USD} \\ \hline
    \multirow{1}{*}{\textbf{initial}}
    & $^{\ast}_{\bullet}\textit{newContest}$ & creat a contest  & \makecell{182909/ \\ 244434}   & \makecell{\$0.53/ \\ \$0.71} \\
    \hline
    \multirow{2}{*}{\textbf{entry}}
    & $\bullet\textit{rootEO}$ & submit EO root  & 45322 & \$0.13 \\
    & $^{\ast}_{\bullet}\textit{submitEO}$ & submit one EO  & 143978 & \$0.42 \\ 
    \hline
    \multirow{4}{*}{\textbf{review}}
    & $\bullet\textit{rootVO}$ & submit VO root & 65956 & \$0.19 \\
    & $^{\ast}_{\bullet}\textit{submitVO}$ & submit one VO  & 62267 & \$0.18 \\ 
    & $\bullet\textit{offChainVerdict}$ & submit results of off-chain verdict  & 44967  & \$0.13 \\
    & $\ast\textit{verdict}$ & directly execute on-chain verdict & \makecell{37227+ \\ 2171$n_d$}  & \makecell{\$0.11+ \\ 0.006$n_d$} \\
    \hline
    \multirow{3}{*}{\textbf{audit}}
    & $\bullet\textit{doubleVotes}$ & report double vote & 65844 & \$0.19 \\
    & $\bullet\textit{reloadChunkVO}$ & reload VO chunks to blockchain & \makecell{37843+ \\ 36578$n_{vo}$} & \makecell{\$0.11+ \\ 0.11$n_{vo}$} \\ 
    & $\bullet\textit{onchainVerdict}$ & redo verdict with smart contract & \makecell{46324+ \\ 2171$n_d$} & \makecell{\$0.14+ \\ 0.006$n_d$} \\
    \hline  
\end{tabular}
\end{center}
\label{t2}
\vspace{-5mm}
\end{table}



\noindent \textbf{\textit{NF-Crowd CC-OCR protocol}}: 
The lower bound of the cost in NF-Crowd CC-OCR also consists of four parts:
(1) a client to set up a new contest via $newContest()$ (\$0.71) during \textit{CC-OCR.initial}, which is more expensive because of the additionally transferred \textXi $deposit$;
(2) the client to upload the Merkle root of $EOs$ via $rootEO()$ (\$0.13) during \textit{CC-OCR.entry};
(3) the client to upload the Merkle root of $VOs$ via $rootVO()$ (\$0.19) during \textit{CC-OCR.review} and finally (4) the client to upload winners via \textit{offChainVerdict\;()} (\$0.13).
The lower bound of the cost is then $\$(0.71+0.13+0.19+0.13) = \$1.16$.
If any misbehavior occurs, the countermeasure functions (i.e., $submitEO()$, $\textit{submitVO}()$, $doubleVotes()$, $\textit{reloadChunkVO}()$, $\textit{onchainVerdict}()$) can be invoked and the cost for calling them will be mainly deducted from \textXi $deposit$ paid by protocol violators.
It is worth noting that the cost of function $\textit{reloadChunkVO}()$, $\$(0.11+0.11n_{vo})$, increases along with the number of $VOs$ carried by a $chunk_{\textit{VO}}$ (i.e., $n_{vo}$).
Grouping $\textit{VOs}$ into $chunks_{\textit{VO}}$ could effectively reduce the cost of reloading $\textit{VOs}$.
For instance, when $n_{\textit{vo}}=1$, the cost of reloading 100 $\textit{VOs}$ would be $\$100*(0.11+0.11*1)=\$22$, which is almost twice of $\$(0.11+0.11*100)=\$11.11$ with $n_{\textit{vo}}=100$.


\begin{figure}
\centering
{
   
    \includegraphics[width=0.78\columnwidth]{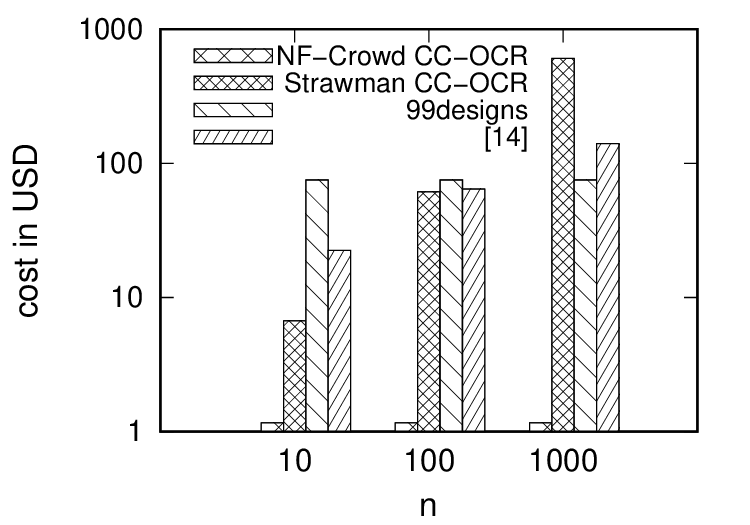}
}
\caption {Cost of the proposed protocols and existing solutions}
\vspace{-3mm}
\label{cost} 
\end{figure}

\subsection{Comparison between NF-Crowd and existing solutions}

Next, we compare the cost of strawman CC-OCR and NF-Crowd CC-OCR with both the platform fee charged by centralized \texttt{99designs} platform and the cost of a recent decentralized protocol in~\cite{duan2019aggregating}.
We assume that the design contest has a \$500 reward.
Also, for the purpose of evaluating the scalability of the solutions, we changed the scale of the crowd $n$ and displayed the results in Fig.~\ref{cost}.
As can be seen, when $n$ is increased from 10 to 1000, the cost of strawman CC-OCR linearly increases from $\$6.7$ to over $\$600$.
In contrast, the cost of NF-Crowd CC-OCR stays at $\$1.16$ constantly regardless of the scale of the crowd.
The cost of \texttt{99designs} stays at \$75, namely 15\% of \$500 reward.
In fact, from $\$8$ reward onwards it is more economical to use NF-Crowd CC-OCR than \texttt{99designs}.
The cost of the decentralized protocol in~\cite{duan2019aggregating} increases from $\$22.4$ to $\$140.5$ when the crowd gets scaled up.
Therefore, we see that NF-Crowd protocols minimize the cost by decoupling the amount of service fees from both the scale of the crowd and the amount of reward.



\section{Related work}
\label{s6}
\subsection{Centralized crowdsourcing}
Crowdsourcing has been emerging as a successful business model and has driven the rise of centralized crowdsourcing platforms such as \texttt{Upwork}~\cite{Upwork}, Amazon Mechanical Turk~\cite{AMT}, \texttt{99designs}~\cite{99designs} and \texttt{designContest}~\cite{designContest}.
Via these platforms, clients publish human intelligence tasks (HIT) that are challenging for computers but easy for human to complete with rewards and interested workers (or freelancers or designers) accomplish the tasks to earn rewards.
As the primary revenue stream, service fees are charged from the rewards by most of these platforms.
\texttt{Upwork}~\cite{Upwork} charges workers a sliding service fee based on the lifetime billings with a specific client, which consists of 20\% of the first \$500, 10\% of the billings between \$500 and \$10000 and 5\% of the billings that exceed \$10000.
For instance, Amazon Mechanical Turk (AMT)~\cite{AMT} charges clients 20\% fee on the reward and bonus amount (if any) clients pay workers.
Currently, Clients and designers are used to accepting the high service fees as they need a relatively trustworthy intermediary to exclude dishonest behaviors~\cite{zhang2012reputation}.
The \texttt{NF-Crowd} protocols proposed in this paper minimize the fees for purchasing trust in crowdsourcing. On the other hand, it makes the off-chain workload heavier. In other words, the \texttt{NF-Crowd} protocols transform the mandatory charge of monetary fees into non-monetary off-chain workload, offering a new option to participants in crowdsourcing.

\subsection{Decentralized crowdsourcing}
Recent advancements in blockchain technologies~\cite{nakamoto2008bitcoin} and smart contract platforms like Ethereum~\cite{wood2014ethereum} are driving the rise of decentralized crowdsourcing systems~\cite{calado2018tamper,duan2019aggregating,feng2019mcs,li2018crowdbc,lu2018zebralancer,wang2018blockchain,wu2019bptm,xu2019blockchain}.
In contrast to centralized crowdsourcing platforms, 
decentralized crowdsourcing systems leverage the distributed miners to decentralize both data storage and computation in a tamper-resilient manner. Thus they eliminate the need for a trusted third party.
However, existing designs of decentralized crowdsourcing systems can hardly handle transaction fees in a scalable manner, resulting in total costs of decentralizing crowdsourcing even higher than service fees charged by centralized crowdsourcing platforms.
For instance, decentralized CrowdBC~\cite{li2018crowdbc} charges 0.011 ether (i.e., \$1.93 by taking the average gas and ether prices in first half of year 2019) to tag 100 images while the same task only spends about \$0.45 in centralized AMT~\cite{AMT}.
In~\cite{duan2019aggregating}, aggregating data from 1,000 data providers costs \$140.
In~\cite{wu2019bptm}, adding each task to the task matching smart contract cost about 210,000 gas (i.e., \$0.61), namely \$610 for adding 1,000 tasks.  
The \texttt{NF-Crowd} protocols proposed in this paper reduce the cost of running decentralized projects on top of Ethereum to a small constant value regardless of the scale of the crowd, which for the first time demonstrates a significant economic advantage in decentralizing crowdsourcing.
In addition, we consider \texttt{NF-Crowd} a generic feature that is orthogonal to other design goals of decentralized crowdsourcing systems, allowing existing protocols to also become \texttt{NF-Crowd} by simply identifying TYPE $n \times 1$ and TYPE $1 \times n$ steps and reducing their cost with our strategies.

\vspace{-1mm}
\subsection{Scaling blockchain with off-chain execution}
Off-chain execution of smart contracts is a promising solution for improving blockchain scalability~\cite{cheng2019ekiden,das2019fastkitten,dziembowski2018general}.
However, recent works in this line have to either assume one honest manager for off-chain execution~\cite{cheng2019ekiden} or allow the execution to get aborted when the manager is dishonest~\cite{das2019fastkitten}.
The state channel network (SCN)~\cite{dziembowski2018general} could achieve the never abort property when at least one participant is honest but it only supports two-participant contracts.
The \texttt{NF-Crowd} protocols proposed in this paper extend the objective to support complex multi-participant multi-round smart contracts without losing the never abort property.
In addition, the \texttt{NF-Crowd} protocols for CC-OCR projects have been implemented in Ethereum, so they are ready-to-use.



\vspace{-1mm}
\subsection{Using cryptocurrency as security deposits}
There have been many recent efforts on blockchain-based protocol design that leverage cryptocurrency as security deposits to penalize unexpected behaviors and improve security~\cite{andrychowicz2014secure,dong2017betrayal,kiayias2015traitor,matsumoto2017ikp}. 
In~\cite{dong2017betrayal}, ether is used as security deposits to provide verifiable cloud computing.
In~\cite{matsumoto2017ikp}, ether is used as security deposits
to enforce certificate authorities to be honest.
Inspired by these previous efforts, \texttt{NF-Crowd} demands each participant to lock ether in smart contracts as security deposits to penalize potential misbehaviors violating the protocol and thereby enforces participants to stay honest.

\section{Conclusion}
\label{s7}
This paper proposes a new suite of protocols called \texttt{NF-Crowd} that 
reliably resolves the scalability issues faced by decentralized crowdsourcing projects. The proposed approach 
reduces the lower bound of the total cost to $O(1)$.
We prove that as long as participants of a project powered by \texttt{NF-Crowd} are rational, the $O(1)$ lower bound of the cost could be reached regardless of the scale of the crowd.
We also demonstrate that as long as at least one participant of a project powered by \texttt{NF-Crowd} is honest, the project cannot be aborted and the results are guaranteed to be correct.
We design \texttt{NF-Crowd} protocols for a representative type of project named crowdsourcing contest with open community review (CC-OCR).
We implement the protocols over the Ethereum official test network.
Our results demonstrate that \texttt{NF-Crowd} protocols can reduce the cost of running a CC-OCR project to less than \$2 regardless of the scale of the crowd, providing a significant cost benefit in adopting decentralized crowdsourcing solutions.


\section*{Acknowledgement}
We thank our shepherd, Alysson Bessani and the anonymous reviewers for their feedback and comments. 
Chao Li acknowledges the partial support by Fundamental Research Funds for the Central Universities (No. 2019RC038).

\renewcommand\refname{Reference}

\bibliographystyle{plain}
\urlstyle{same}

\bibliography{main.bib}

\vskip 0pt plus -1fil

\begin{IEEEbiography} [{\includegraphics[width=1in,height=1.25in]{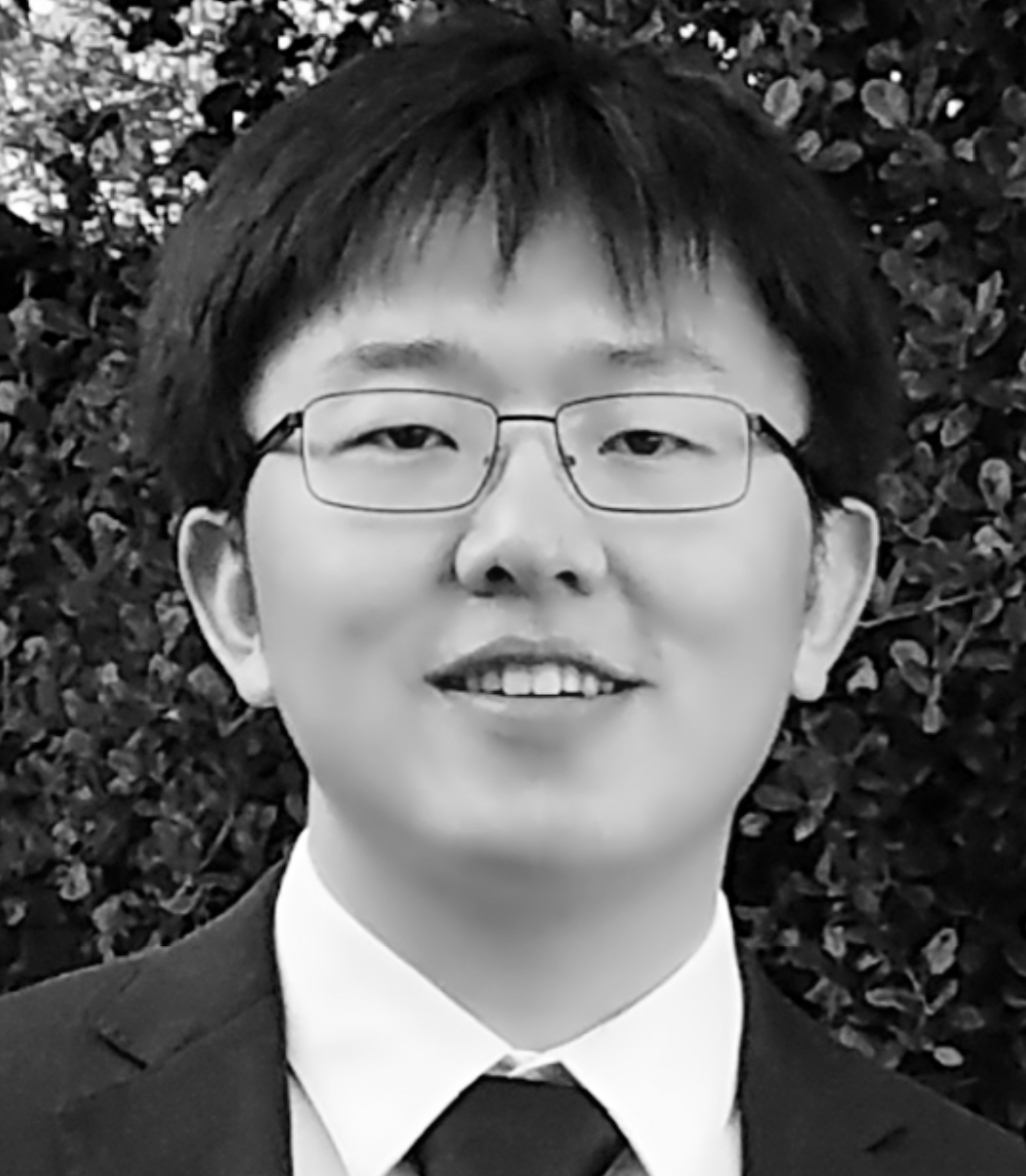}}]{Chao Li}
    is an Assistant Professor in the School of Computing and Information Technology at Beijing Jiaotong University. He received his Ph.D. degree from the School of Computing and Information at University of Pittsburgh and his MSc degree from Imperial College London. His current research interests are focused on Blockchain and Data Privacy. 
\end{IEEEbiography}

\vskip 0pt plus -1fil

\begin{IEEEbiography}
    [{\includegraphics[width=1in,height=1.25in]{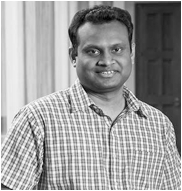}}]{Balaji Palanisamy}
is an Associate Professor in the School of computing and information in University of Pittsburgh. He received his M.S and Ph.D. degrees in Computer Science from the college of Computing at Georgia Tech in 2009 and 2013, respectively. His primary research interests lie in scalable and privacy-conscious resource management for large-scale Distributed and Mobile Systems. At University of Pittsburgh, he codirects research in the Laboratory of Research and Education on Security Assured Information Systems (LERSAIS).
\end{IEEEbiography}

\vskip 0pt plus -1fil

\begin{IEEEbiography} [{\includegraphics[width=0.9in,height=1.25in]{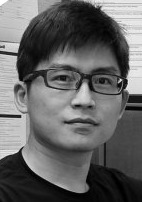}}]{Runhua Xu}
     is a PhD student of School of Computing and Information at
the University of Pittsburgh. He received his M.S. in Computer Science
and B.E. in Software Engineering degrees from Beihang University and
Northwestern Polytechnical University, China, respectively. His research
interests include Access Control, Applied Cryptography, Data Security
and Privacy.
\end{IEEEbiography}

\vskip 0pt plus -1fil

\begin{IEEEbiography} [{\includegraphics[width=1in,height=1.25in]{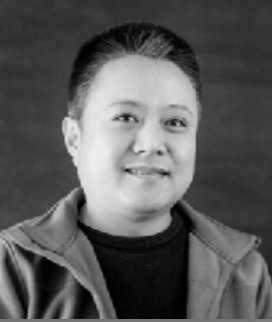}}]{Jian Wang}
    is an Associate Professor in the School of Computing and Information Technology at Beijing Jiaotong University. He received his Ph.D. degree from Beijing University of Posts and Telecommunications in 2008. His current research interests are focused on Network Security and Software Security.
\end{IEEEbiography}

\vskip 0pt plus -1fil

\begin{IEEEbiography} [{\includegraphics[width=1in,height=1.25in]{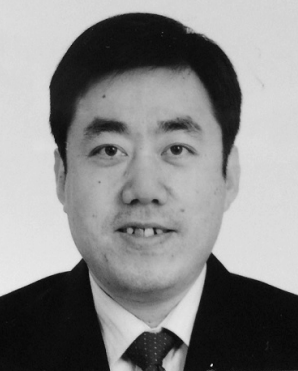}}]{Jiqiang Liu}
    received the B.S. and Ph.D. degrees
from Beijing Normal University, Beijing, China, in
1994 and 1999, respectively. He is currently a Professor with the School of Computer and Information
Technology, Beijing Jiaotong University, Beijing. He
has authored over 60 scientific papers in various
journals and international conferences. His main
research interests are trusted computing, cryptographic protocols, privacy preserving, and network
security.
\end{IEEEbiography}

\end{document}

%% file: solidity-highlighting.tex
\usepackage{listings, xcolor}

\definecolor{verylightgray}{rgb}{.97,.97,.97}

\lstdefinelanguage{Solidity}{
	keywords=[1]{anonymous, assembly, assert, balance, break, call, callcode, case, catch, class, constant, continue, constructor, contract, debugger, default, delegatecall, delete, do, else, emit, event, experimental, export, external, false, finally, for, function, gas, if, implements, import, in, indexed, instanceof, interface, internal, is, length, library, log0, log1, log2, log3, log4, memory, modifier, new, payable, pragma, private, protected, public, pure, push, require, return, returns, revert, selfdestruct, send, solidity, storage, struct, suicide, super, switch, then, this, throw, transfer, true, try, typeof, using, value, view, while, with, addmod, ecrecover, keccak256, mulmod, ripemd160, sha256, sha3}, 
	keywordstyle=[1]\color{blue}\bfseries,
	keywords=[2]{address, bool, byte, bytes, bytes1, bytes2, bytes3, bytes4, bytes5, bytes6, bytes7, bytes8, bytes9, bytes10, bytes11, bytes12, bytes13, bytes14, bytes15, bytes16, bytes17, bytes18, bytes19, bytes20, bytes21, bytes22, bytes23, bytes24, bytes25, bytes26, bytes27, bytes28, bytes29, bytes30, bytes31, bytes32, enum, int, int8, int16, int24, int32, int40, int48, int56, int64, int72, int80, int88, int96, int104, int112, int120, int128, int136, int144, int152, int160, int168, int176, int184, int192, int200, int208, int216, int224, int232, int240, int248, int256, mapping, string, uint, uint8, uint16, uint24, uint32, uint40, uint48, uint56, uint64, uint72, uint80, uint88, uint96, uint104, uint112, uint120, uint128, uint136, uint144, uint152, uint160, uint168, uint176, uint184, uint192, uint200, uint208, uint216, uint224, uint232, uint240, uint248, uint256, var, void, ether, finney, szabo, wei, days, hours, minutes, seconds, weeks, years},	
	keywordstyle=[2]\color{teal}\bfseries,
	keywords=[3]{block, blockhash, coinbase, difficulty, gaslimit, number, timestamp, msg, data, gas, sender, sig, value, now, tx, gasprice, origin},	
	keywordstyle=[3]\color{violet}\bfseries,
	identifierstyle=\color{black},
	sensitive=false,
	comment=[l]{//},
	morecomment=[s]{/*}{*/},
	commentstyle=\color{gray}\ttfamily,
	stringstyle=\color{red}\ttfamily,
	morestring=[b]',
	morestring=[b]"
}